\documentclass[12pt]{article}
\pdfoutput=1 
\usepackage{amsfonts}
\usepackage{mathrsfs}
\usepackage{amsthm}
\usepackage{bbding} 
\usepackage{amssymb}
\usepackage{amsmath}
\usepackage{commath}
\usepackage{graphicx,color,epstopdf}
\usepackage{enumitem}
\usepackage[colorlinks,linkcolor=red,anchorcolor=green,citecolor=blue]{hyperref}
\usepackage{extarrows}
\usepackage{booktabs,multirow} 
\usepackage{natbib} 
\usepackage[usenames,dvipsnames,svgnames,table]{xcolor} 

\numberwithin{equation}{section}

\newtheorem{thm}{Theorem}[section]

\newtheorem{lem}{Lemma}[section]
\newtheorem{rem}{Remark}

\begin{document}
\cleardoublepage
\pagenumbering{gobble}

\date{}

\title{Specification testing for regressions: an approach bridging between local smoothing and  global smoothing methods}
\author{Lingzhu Li$^{1}$ and Lixing Zhu$^{1, 2}$\footnote{The corresponding author.
Email: lzhu@hkbu.edu.hk. The research described herewith was supported by a grant from the University Grants Council of Hong Kong and a NSFC grant (NSFC11671042).}\\
  \small $^1$Department of Mathematics, Hong Kong Baptist University, Hong Kong \\
  \small $^{1, 2}$School of Statistics, Beijing Normal University, Beijing, China \\
}

\maketitle

\renewcommand\baselinestretch{1.4}
{\small

\noindent {\bf Abstract:}
For regression models, most of existing specification tests can be categorized into  the class of  local smoothing tests and of global smoothing tests.
Compared with global smoothing tests, local smoothing tests can only detect local alternatives distinct from the null hypothesis at a much   slower rate when the dimension of predictor vector is high, but can be more sensitive to oscillating alternatives.
In this paper, we suggest a projection-based test to  bridge between the local and global smoothing-based methodologies such that the test can benefit from the advantages of these two types of tests.
The test construction  is based on a kernel estimation-based
 method and the resulting test becomes a distance-based  test with a
 closed form. The asymptotic properties are investigated.
Simulations and a real data analysis are conducted to evaluate the performance of the test in finite sample cases.

\bigskip


}

\newpage

\cleardoublepage
\pagenumbering{arabic}

\newpage
\setcounter{equation}{0}
\section{Introduction}\label{sec: introduction}

Statistical inference and prediction are of great interest and importance to decision-making in numerous fields such as medicine area, bioinformatics and econometrics.
The foundation of statistical analysis is statistical models.
{Among various statistical models, parametric ones are widely used} because  statistical analysis can then be more efficient if the model structure is proper.
{Since the parameter is unknown}, 
we need to first get an estimation to conduct further analysis.
However, a wrongly specified model would result in unreliable estimations and following statistical inferences.
Therefore, it is important to test the model structure before a model is applied in any further regression analysis.
 Suppose the null model is of the following form:
\begin{eqnarray}\label{model: H_0}
  Y = g(X,\theta)+\varepsilon
\end{eqnarray}
where $g$ is a known function, $\theta \in\Theta\subset R^d$ is the unknown parameter and $E(\varepsilon|X) = 0$.
$X$ is a random vector in $R^p$ as the predictor and $Y$ is the response variable in $R$.
To check the adequacy of the model \eqref{model: H_0}, consider a general alternative model
\begin{eqnarray}\label{model: H_1}
  Y = G(X)+\varepsilon
\end{eqnarray}
where $G(\cdot) \notin \{g(\cdot,\theta): \theta\in \Theta\} $ is an unknown function.

{There are many proposals available in the literature.}
But how to efficiently deal with high-dimensional data is always a concern. A frequently used methodology is to transform the problem to a problem
at all projection directions. To be precise, test statistic can be based on, say, univariate projected predictors $\alpha^{\tau}X$ for all $\alpha$ in a subset of $R^p$. This is an idea of projection pursuit regression that was proposed by
\cite{friedman1981projection}.
\cite{huber1985projection} is a comprehensive reference of this methodology. Along this line,
\cite{bierens1990consistent} furthered the method by using the Fourier transformation that gives the weight functions to be $\exp(it^{\top} X)$ for all $t \in R^p$.
The integration over $t$ with respect to a measure can then formulate a final test statistic. This test is of a dimension reduction nature as every weight function uses univariate projected predictor $t^{\top}X$.
\cite{hongzhi1992testing},
\cite{lixing1998dimension},
\cite{escanciano2006consistent}, 
\cite{stute2008model} and
\cite{lavergne2008breaking} and 
\cite{lavergne2012one} are the relevant references in this field.
To be precise,
\cite{escanciano2006consistent} proposed an omnibus test by using a residual marked empirical process whose index set contains  all projection directions.
The test in 
\cite{lavergne2008breaking}
is also based on the empirical process, but the integral over all directions leads to a simple closed form of the test statistic.
\cite{stute2002model} is also based on a residual marked empirical process, but its index set contains only one projection direction.
\cite{stute2008model} used a predictor-marked residual process to construct a test.
\cite{lavergne2012one} developed a smooth integral conditional moment test constructed by
\cite{zheng1996consistent} that uses nonparametric kernel estimation of some conditional moment.
\cite{guozhu2017review} is a comprehensive review.

The aforementioned tests can be categorized into two very different classes:  nonparametric estimation-based and empirical process-based. Then they can be classified as local smoothing and global smoothing methods. This is because nonparametric estimation-based methods rely on local smoothing techniques and empirical process-based tests are the averages of functions of weighted sum of residuals over an index set, which is a global smoothing step.
\cite{lixing1998dimension} and
\cite{lavergne2012one}, are based on nonparametric  estimation for the conditional moment and thus belongs to the class of local smoothing methods. \cite{guo2016model} introduced an adaptive-to-model test that is based on \cite{zheng1996consistent}.
The others are constructed in a global smoothing manner, such as
\cite{stute2002model},
\cite{escanciano2006consistent} which are based on  empirical processes.
Other examples in the class of global smoothing tests include
\cite{zhu2003model},
\cite{tan2016projection}.
The tests are called global smoothing tests because nonparametric estimation is avoided  and  global averages over a group of statistics indexed by a set of indices is formulated as  final test statistics.

These two classes of tests have their own pros and cons, which have been discussed frequently in the literature. If we do not use projected predictors, but the original $p$-dimensional predictor $X$, the inefficiency of  nonparametric estimation in high-dimension cases cause local smoothing tests to hardly maintain the significance level and dramatically lose its power  as the dimension $p$ increases. They can only detect the local alternatives distinct from the null at the rate of order $n^{-1/2}h^{-p/4}$, where $h$ is the bandwidth going to zero as $n\to \infty$.
Some methods require some dimension reduction model structures under either the null or the alternatives.
For instance, 
\cite{guo2016model} designed an adaptive-to-model test for the single-index model: $Y = g(\beta^{\top} X)+\varepsilon$ where $g$ is a known function and $\beta$ is the unknown parameter. The test can detect the local alternatives converging to the null at the rate of $n^{-1/2}h^{-1/4}$.
For null models that have $q$ projection directions with a $p\times q$ matrix $\beta$, this rate slows down to $n^{-1/2}h^{-q/2}$.
To alleviate the negative impact from the dimensionality, the projection-based tests work well. The test in
\cite{lavergne2012one} is a local smoothing test, but can detect the local alternatives distinct from the null at the rate of $ n^{-1/2}h^{-1/4}.$ It is worth noticing that it is still a local smoothing test. As $h \to \infty$, this rate must slower than $n^{-1/2}$.
In contrast, global smoothing tests can always detect local alternatives distinct from the null at the fastest possible rate that is   $n^{-1/2}$.
For global smoothing tests, the local alternatives distinct from the null at the rate of  $1/\sqrt n$  can be detected.
\cite{Delgado2017asymp}
proved that some global smoothing tests such as
\cite{stute1997} and 
\cite{stute1998} have asymptotic optimality including asymptotically uniformly most powerful in a semiparametric context and asymptotically semiparametric efficient respectively. But
 many of them do not have  tractable limiting null distributions.
This requires using  re-sampling methods to determine critical values, such as either the bootstrap or the wild bootstrap or the Monte Carlo approximation, to approximate the corresponding  sampling null distributions.

In this paper, we propose a projection-based specification test. Like any projection-based test such as
\cite{escanciano2006consistent} and
\cite{lavergne2012one}, 
we project the predictor onto one-dimensional subspaces such that at any direction, the test only involves univariate predictor. However, the key feature of the proposed test distinguishing from  these existing projection-based tests is that the proposed test  bridges between local and global smoothing methodology. The resulting test can have a simple closed form and  the advantages of global smoothing test as we discussed above although it is based on a local smoothing test. Thus, it could benefit from both.

The rest of this paper is organized as follows.
In Section \ref{sec: method}, the test statistic  construction is described.
Section \ref{sec: asymptotic properties} presents the asymptotic properties  under the null  and  alternative hypothesis.
In section \ref{sec: numerical analysis}, numerical studies are reported, including  simulations and a real data analysis. The results indicate that  the proposed test does benefit from both local and global smoothing testing procedures.
Section \ref{sec: discussion} contains some discussions.
Technical proofs are postponed to  the Appendix.

\setcounter{equation}{0}
\section{Test statistic construction}\label{sec: method}
\subsection{Basic idea}
From the models we stated in the previous section, the hypotheses are as follows:
\begin{eqnarray}\label{test hypothesis}
  H_0 &:& \mathbb{P}(E(Y|X) = g(X,\theta_0))=1 \text{ for some } \theta_0 \in \Theta \subset R^d, \nonumber \\
  H_1 &:& \mathbb{P}(E(Y|X) = g(X,\theta))<1 \text{ for all } \theta \in \Theta
\end{eqnarray}
where $ g(\cdot) $ is a known regression function and $\theta_0$ is the unknown parameter vector.
Define $e = Y-g(X,\theta_0) $ as the residual at the population level.
Under the null hypothesis $e = \varepsilon$  with the condition $E(e|X)=0$.
Let $f(\cdot)$ and $f_\alpha(\cdot)$ be respectively the density function of $X$ and $\alpha^{\top} X$.
Notice that $E(e|X)=0 $ holds if and only if $E(E^2(e|X)f(X))=0$ under  some continuous conditions on $f(\cdot)$, see
\cite{zheng1996consistent}.
Further, notice that $E(E^2(e|X)f(X))=0$ is equivalent to
\begin{eqnarray}
    E(E^2(e|\alpha^{\top} X)f_\alpha(\alpha^{\top} X))=0 \text{  for all }
    \alpha\in R^p  \nonumber
\end{eqnarray}
where $f_\alpha(\cdot)$ and $\mu(\cdot) $ are respectively the conditional density function of $\alpha^{\top} X$ when $\alpha$ is given and the marginal density function of  $\alpha$.
The following is a slightly extension of { Lemma 1}
of \cite{escanciano2006consistent} in which
the projection direction $\alpha$ is limited to the unit hypersphere $S^p = \{ \alpha: \|\alpha\|=1 \}$.
It can be checked that
\begin{eqnarray*}
  E(e|\alpha^{\top} X) = 0 \iff E(e|c\cdot \alpha^{\top} X)=0 \text{ for all } c\in R.
\end{eqnarray*}
Thus we obtain the following lemma.
\begin{lem}\label{equivalent condition}
Suppose $\eta$ is a random variable such that $E|\eta|<\infty$ and $\xi\in R^p$ is a random vector.
Then
    $E(\eta|\xi)=0$ holds if and only if $E(\eta|\beta^{\top} \xi) = 0 $ holds for all $\beta\in R^p$. Further, assume that $f_\alpha(\cdot)$ and  $\mu(\cdot)$ are positive on their supports, $E(e|X) = 0 $ almost surely holds if and only if
  $\int_{R^p} E(E^2(e|\alpha^{\top} X)f_\alpha(\alpha^{\top} X)) \mu(\alpha)d \alpha = 0. $
\end{lem}

Under the alternative hypothesis, { Lemma 2.1} implies that there exists at least an $\alpha^\ast\in R^p$ such that $E(e|\alpha^\ast)\neq 0$, and then $\int_{R^p} E(E^2(e|\alpha^{\top} X)f_\alpha(\alpha^{\top} X))\mu(\alpha) d \alpha > 0$.
Therefore, We can use an estimator of this quantity to construct a test statistic.

Suppose we have an i.i.d.  sample $\{(x_i,y_i)\}_{i=1}^n $ from $(X,Y)$. The least squares estimate of $\theta_0$ is defined as  $\hat\theta_n = \arg\min\limits_{\theta\in \Theta} \frac{1}{n}\sum_{i=1}^n (y_i-g(x_i,\theta))^2 $.
Let $\hat e_i = y_i - g(x_i,\hat\theta_n) $ be the residual at the sample level.
Under some regularity conditions, $\hat\theta_n$ is a consistent estimate of $\theta_0$ under the null hypothesis and of an $\theta \in \Theta$ under the alternative hypothesis. Throughout the rest of this paper, we will not list the detailed conditions. The readers can refer to
\cite{white1981consequences} (Corollary 2.2) and
\cite{bierens1982consistent} (Theorem 9).

\subsection{Test construction}
To start with the construction, we review two existing tests first.
\cite{zheng1996consistent}'s test is an empirical version of $E(eE(e|X)f(X))$ as follows:
\begin{eqnarray*}
  \frac{1}{n(n-1)}\sum_{i=1}^{n}\sum_{i\neq j}
  \hat e_i \hat e_j \frac{1}{h^p} K_p(\frac{x_i-x_j}{h})
\end{eqnarray*}
where $K_p(\cdot)$ is a product  kernel function and $h$ is the bandwidth.
With some regularity conditions, the test statistic multiplying $nh^{p/2}$ goes to its weak limit under the  null where $h\to 0$
as $n\to \infty$.
\cite{lavergne2012one}'s test is an integrated \cite{zheng1996consistent}'s test over all projection directions $\alpha \in S^p$. It has the formula as
\begin{eqnarray*}
  \frac{1}{n(n-1)}\sum_{i=1}^n\sum_{j\neq i} \hat e_i\hat e_j \int_{B} \frac{1}{h}K\left(\frac{\alpha^{\top}(x_i-x_j)}{h}\right) d \alpha
\end{eqnarray*}
where $B$ is $S^p$ or can be some subset of $S^p$.
This test greatly alleviates the curse of dimensionality as it multiplied by $nh^{1/2}$ tends to its weak limit under the null. In their construction, the projection direction $\alpha$ is assumed to be uniformly distributed. It is noted that it is still a local smoothing test as the integral $\int_{B} \frac{1}{h}K\left(\frac{\alpha^{\top}(x_i-x_j)}{h}\right) d \alpha$ still involves the bandwidth $h$ and the convergence rate $nh^{1/2}$ is still slower than the rate $n$ when a quadratic form of global smoothing test is used such as 
\cite{stute2002model}.
Also, the integral does not have a closed form and then the computation is an issue when the dimension $p$ is high. \cite{lavergne2012one} used a Monte Carlo approximation for this integral. The computation is time-consuming in high-dimensional scenarios.

We now modify their construction to derive our test statistic.
First, use the kernel estimate to replace the conditional moment $E(e|\alpha^{\top} X)$ and the density function $f_\alpha(\alpha^{\top} x)$ of $\alpha^{\top} X$ as
\begin{eqnarray*}
  \hat f_\alpha(\alpha^{\top} x_i) &=& \frac{1}{n-1}\sum_{j\neq i}\frac{1}{h}K\left(\frac{\alpha^{\top} (x_i-x_j)}{h}\right), \\
  \hat E(e|\alpha^{\top} x_i) &=& \frac{\frac{1}{n-1}\sum_{j\neq i}^n \hat e_j \frac{1}{h}K\left(\frac{\alpha^{\top} (x_i-x_j)}{h}\right)}
  {\hat f_\alpha(\alpha^{\top} x_i)}
\end{eqnarray*}
where $K(\cdot)$ is the kernel function and $h$ is the smooth parameter.
 Note that \[
E(E^2(e|\alpha^{\top} X)f_\alpha(\alpha^{\top} X)) = E(eE(e|\alpha^{\top} X)f_\alpha(\alpha^{\top} X)).
\]
Thus $V \overset{\Delta}{=} \int_{R^p}E(E^2(e|\alpha^{\top} X)f_\alpha(\alpha^{\top} X))\mu(\alpha)d \alpha$ can be estimated by
\begin{eqnarray}\label{hat V}
  \hat V &\overset{\Delta}{=}& \int_{R^p} \frac{1}{n(n-1)}\sum_{i=1}^n\sum_{j\neq i} \hat e_i\hat e_j \frac{1}{h}K\left(\frac{\alpha^{\top}(x_i-x_j)}{h}\right) \mu(\alpha) d \alpha \nonumber \\
  &=& \frac{1}{n(n-1)}\sum_{i=1}^n\sum_{j\neq i} \hat e_i\hat e_j \int_{R^p} \frac{1}{h}K\left(\frac{\alpha^{\top}(x_i-x_j)}{h}\right) \mu(\alpha) d \alpha .
\end{eqnarray}
We can use this quantity to be a test statistic. Note that it involves the integral and seems still a local smoothing test.
We now choose some particular kernel function $K(\cdot)$ and measure $\mu (\cdot)$ to derive a statistic  that has a closed form. We consider Gaussian kernel and assume that the measure $\mu$ is also Gaussian. To be precise,

Let
\begin{eqnarray}
  K(u) = (2\pi)^{-1/2}\exp(-u^2/2),
\end{eqnarray}
 and consider $\alpha\sim N(0,\sigma^2 I_p)$ where $\sigma^2$ is a variance function of $\alpha$ to be determined, $I_p$ is an identity matrix of dimension $p$. The density function  $\mu$ is
\begin{eqnarray}
  \mu(\alpha) &=& (2\pi)^{-p/2} |\sigma^2 I_p|^{-1/2} \exp\left( -\frac{\alpha^{\top} (\sigma^2 I_p)^{-1}\alpha}{2} \right) \nonumber \\
  &=& (2\pi)^{-p/2} \sigma_\alpha^{-p} \exp \left( -\frac{\alpha^{\top} \alpha}{2\sigma^2} \right). \nonumber
\end{eqnarray}
Thus,  we have the following lemma.
\begin{lem}
When the above Gaussian kernel and measure $\mu$ are used, we have
\begin{eqnarray}
  \hat V &=& \frac{1}{n(n-1)}\sum_{i=1}^n\sum_{j\neq i} \hat e_i\hat e_j
  h^{-1} (2\pi)^{-1/2} \sigma^{-p}
  (\frac{d_{ij}}{h^2}+\frac{1}{\sigma^2} )^{-1/2} \sigma^{p-1} \nonumber\\
  &=& \frac{1}{\sqrt{2\pi}} \cdot \frac{1}{n(n-1)}\sum_{i=1}^n\sum_{j\neq i} \hat e_i\hat e_j
  \frac{1}{\sqrt{\sigma^2 d_{ij}+h^2}}.
\end{eqnarray}
When $\sigma^2$ is chosen to be $h^2$, we have
\begin{eqnarray}\label{1/h times V_n}
   \hat V = \frac{1}{h\sqrt{2\pi}}\cdot \frac{1}{n(n-1)}\sum_{i=1}^n\sum_{j\neq i} \hat e_i\hat e_j \frac{1}{\sqrt{d_{ij}+1}}.
\end{eqnarray}
where $d_{ij}=\|x_i-x_j \|^2 = (x_i-x_j)^{\top} (x_i-x_j).$
\end{lem}

The ``kernel function" $1/(\sqrt{d_{ij}+1})$ in the new formula does not involve the bandwidth $h$ and the quantity $h$ outside the sum can be leave out from the test statistic, also free of the integration. The resulting test statistic is finally defined as
\begin{eqnarray}\label{V_n:noh}
    V_n &=& \frac{1}{n(n-1)}\sum_{i=1}^{n}\sum_{i\neq j} \hat e_i \hat e_j \frac{1}{\sqrt{d_{ij}+1}}
\end{eqnarray}
where $d_{ij} = \|x_i-x_i\|^2 $.

Note that this test is of the structure of a global smoothing test although it is based on projection and a local smoothing test.
Therefore this projection-based test indeed  bridges between local and global smoothing test.

\begin{rem}
  Here we choose the Gaussian kernel function.
  In fact, for any kernel function $K(\cdot)$ and $\alpha\sim N(0,h^2I_p)$, it is easy to see that
  \begin{eqnarray*}
    && \int_{R^p} K\left(\frac{\alpha^{\top} (x_i-x_j) }{h}\right)\exp(-\frac{\alpha^{\top} \alpha}{2 h^2}) d \alpha \\
    &=& h^p \int_{R^p} K\left(t^{\top} (x_i-x_j)\right)\exp(-t^{\top} t/2) dt .
  \end{eqnarray*}
  The corresponding test based on this integral is equivalent to a global smoothing test since the bandwidth $h$ plays no role in the resulting kernel and then can be left out.
  Besides, from the property of kernel function, we can know the resulting test is just based on the distance $\|x_i-x_j\|$ and the concomitant residuals $\hat e_i$ and $\hat e_j$. But this integral may not always have a closed form and thus, computation might be a concern.
\end{rem}

\setcounter{equation}{0}
\section{Asymptotic properties}\label{sec: asymptotic properties}
  Introduce some notations first.
Let
\begin{eqnarray*}
  \dot g(X,\theta) = \frac{\partial g(X,\theta)}{\partial \theta}.
\end{eqnarray*}
For notational simplicity, write $\dot g_i$ as $\dot g(x_i,\theta_0)$.
Define
\begin{eqnarray*}
  H_{\dot g} = E(\dot g\dot g^{\top})
\end{eqnarray*}
and assume it is a nonsingular matrix.
Other notations are:
$$w_{ij} = \frac{1}{\sqrt{d_{ij}+1}} , \, \, \, E_{1i} = E(\dot g_j w_{ij}|x_i) , $$
and
$$ \tilde w_{ij} = w_{ij}
- 2 \dot g_j^{\top} H_{\dot g}^{-1} E_{1i}
+ \dot g_i^{\top} H_{\dot g}^{-1}  E(\dot g_k E_{1k}^{\top}) H_{\dot g}^{-1} \dot g_j . $$

\subsection{Asymptotics under the null hypothesis}
To get the asymptotic properties under the null hypothesis, we use  U-statistics theory. Note that
$V_n$ is an U-statistic as
\begin{eqnarray*}
  U_n = \frac{1}{n(n-1)}\sum_{i=1}^n\sum_{j\neq i} \varepsilon_i \varepsilon_j h(x_i,x_j)
\end{eqnarray*}
where $h(x_i,x_j) = \frac{1}{2}(\tilde w_{ij} +\tilde w_{ji} )$.

Here we introduce some important quantities.
Let $\lambda_1 \ge \lambda_2 \ge \lambda_3 \ge \dots$ be the corresponding eigenvalues to the  solutions $s_1,s_2,\dots$ satisfying that
\begin{eqnarray*}
  \int_{-\infty}^\infty h(x_i,x_j)s(x_j)dF(x_j) = \lambda s(x_i)
\end{eqnarray*}
where $F(\cdot)$ is the distribution function of $X$.
Define
\begin{eqnarray*}
  \mu_0 = E(\sigma^2_i\dot g_k^{\top} H_{\dot g}^{-1}\dot g_i\dot g_i^{\top} H_{\dot g}^{-1} E_{1k} ) - 2 E(\sigma_i^2 \dot g_i^{\top} H_{\dot g}^{-1} E_{1i} )
\end{eqnarray*}
where $\sigma_i^2 = E(\varepsilon^2|x_i)$.

The limiting null distribution of $T_n \overset{\Delta}{=}n V_n$ is stated in the following theorem.
\begin{thm}\label{theorem: null}
Under the null hypothesis  with the regularity conditions in the Appendix,
  \begin{eqnarray*}
  T_n \overset{d}{\to} \sum_{i=1}^\infty \lambda_i (Z_i^2-1) + \mu_0
\end{eqnarray*}
where $\overset{d}{\to}$ stands for convergence in distribution and $Z_1, Z_2, \dots$ are independent standard normal random variables. If $\sigma_\varepsilon^2 \equiv E(\varepsilon^2|X)$ is a constant free of $X$, then
  \begin{eqnarray*}
    \mu_0 = -\sigma_\varepsilon^2 E(\dot g_i^{\top} H_{\dot g}^{-1} E_{1i}) .
  \end{eqnarray*}
\end{thm}

\begin{rem}
 This theorem shows that the limiting null distribution is intractable and thus a Monte Carlo approximation is necessary. In the numerical studies, we use the wild bootstrap to implement the testing procedure.
 \end{rem}

\subsection{Power study}
Suppose the sample $\{ (x_i,y_i) \}_{i=1}^n$ is from the following sequence of models
\begin{eqnarray}\label{model: general}
  Y = g(X,\theta_0)+\delta_n\ell(X)+\varepsilon .
\end{eqnarray}
The values $\delta_n\to 0$ correspond to the local alternative models, fixed nonzero $\delta_n$ to the global alternative model and $\delta_n=0$ to the null model.
Let $\ell_j = \ell(x_j)$ and $M_i = E\left(\ell_j(\tilde w_{ij}+\tilde w_{ji})|X_i=x_i\right) $ for notational simplicity.
\begin{thm}\label{theorem: alternative}
 With the regularity conditions in the {\it Appendix},
  \begin{enumerate}
    \item Under the global alternative model with a fixed nonzero $\delta_n$, in probability,
    \begin{eqnarray*}
      T_n/n \overset{p}{\to} \mu_1
    \end{eqnarray*}
    where $\mu_1 = E(\delta_n^2 \ell_i \ell_j w_{ij}) $.

   \item Under the local alternative model with $\delta_n\to 0$ and $\sqrt n \delta_n\to\infty $ as $n\to\infty$,
  	\begin{eqnarray*}
  		T_n/(n \delta_n^2) \overset{p}{\to} \Delta_\mu ,
  	\end{eqnarray*}
   where  $\Delta_\mu = E(\ell_i\ell_j \tilde w_{ij})  $.
    \item Particularly, under the local alternative model with $\delta_n = n^{-1/2}$,
    \begin{eqnarray*}
      T_n \overset{d}{\to} N(\mu_{1n},\Sigma),
    \end{eqnarray*}
   where $\mu_{1n} = \Delta_\mu+\mu_0$ and $\Sigma = Var(\varepsilon_i M_i)$.
  	  \end{enumerate}
\end{thm}
\begin{rem}
  This theorem shows that the test behaves like a global smoothing test although it is based on the Zheng's test with projected predictors.
\end{rem}

\setcounter{equation}{0}
\section{Numerical studies}\label{sec: numerical analysis}
\subsection{Simulations}
To study the performance of our test, we conduct some simulations under different model settings.
In scenario 1, the dataset is generated from a sequence of models that are oscillating/high-frequent under the alternatives; correspondingly in scenario 2, the dataset is from a sequence of models that are low-frequent under the alternatives; in scenario 3, we study the impact of correlation between the components of $X$ to our test.
For each scenario, we also investigate the influence of the dimension $p$ to the competitors.
From scenario 1 to scenario 3, the null models are  linear. So in scenario 4, we consider a nonlinear model as the null model.
Scenario 1 is designed to exam our test under the oscillating alternative models which usually are in favor of local smoothing tests.
We then compare our test with a typical local smoothing test:
\cite{zheng1996consistent}'s test.
As the null model is linear, which is under the single-index framework, we then also consider \cite{guo2016model}'s test as a competitor.
Another competitor is
\cite{stute1998bootstrap}'s test since it is a typical global smoothing test.
Our test is denoted as $T_n$  and
\cite{guo2016model}'s test, \cite{stute1998bootstrap}'s test and \cite{zheng1996consistent}'s test are denoted as $T^{GWZ}$, $T^{S}$ and $T^{Zh}$ respectively.
For the kernel estimation in \cite{guo2016model}'s test and \cite{zheng1996consistent}'s test, the choices of the bandwidth $h$ are the same as those in \cite{guo2016model}.
The critical values for our test are  the $95\%$ quantile of $300$ wild bootstrap samples.

{\it Scenario 1. } Consider
\begin{eqnarray*}
  Y &=& \beta^{\top} X + a\cdot cos(\beta^{\top} X)+ \varepsilon .
\end{eqnarray*}
$a=0$ corresponds to the null hypothesis and $a\neq 0$ to the alternative hypothesis.
To examine the power performance,
$a = 0.2, 0.6, 1$. 
The parameter $\beta = (1,1,\dots,1)/\sqrt p$, $\|\beta\|=1$.
The predictors $x_i,i=1,\dots,n$ are independently generated from the multivariate normal distribution $N(0,I_p)$.
The errors $\varepsilon_i, i=1,\dots,n$ are independently drawn from the standard normal distribution $N(0,1)$.
The dimension $p = 2, 4, 8$ and the sample size $n = 200$.
We conducted 1000 experiments for each scenario.
The empirical size and powers are presented in Table \ref{tab: Scenario 1}.

\begin{center}
Table \ref{tab: Scenario 1} about here.
\end{center}

The results show that our test can maintain the significance level well for both dimensions $p=2$ and $p=8$ while $T^S$ and $T^{Zh}$ cannot control type I error under the case $p=8$.
When $a$ increases with larger deviation from the null, the powers reasonably increase for all the competitors, but $T_n$ surpasses $T^S$ and $T^{Zh}$ in both settings with the dimension $p=2$ and $p=8$ and is comparable to $T^{GWZ}$. These findings suggest that the proposed test $T_n$ can have good performance for the oscillating alternative model although it is a global smoothing test. In other words, it does benefit the merit of local smoothing test. We also note that the adaptive-to-model test $T^{GWZ}$ works well slightly better than  $T_n$. This is because $T^{GWZ}$ fully uses the dimension reduction structure under the null and is also a local smoothing test. Compared with the global smoothing test $T^S$, $T_n$ performs much better. Further, The local smoothing test $T^{Zh}$ clearly suffers from the data sparseness in high dimensional space.

Next, we study the tests performances under a low-frequency model.

{\it Scenario 2. } Consider
\begin{eqnarray*}
  Y &=& \beta_1^{\top} X + a\cdot 0.3 (0.5+\beta_2^{\top} X)^3 + \varepsilon .
\end{eqnarray*}
In this scenario, we test against a low-frequency alternative model.
The parameters are $\beta_1 = (\underbrace{1,\dots,1}_{p/2},0,\dots,0)^{\top} /\sqrt{p/2}$ and $\beta_2 = (\underbrace{0,\dots,0}_{p/2},1,\dots,1)^{\top} /\sqrt{p/2} $.
The sample size $n=200$ and $p = 2, 8$.
$X$ and $\varepsilon$ follow the same distribution in {\it Scenario 1}, i.e. $X\sim N(0,I_p)$ and $\varepsilon\sim N(0,1) $.
We use $a = 0.2,  0.6,  1$ under the alternative models.
The plots of the power line is shown in Figure \ref{fig: scenario 2}.

\begin{center}

Figure \ref{fig: scenario 2} about here.
\end{center}

The results give us the following observations. When $p=2$, the global smoothing test $T^S$ works well and $T_n$ performs similarly or very slightly worse compared with $T^S$. Two local smoothing tests $T^{GWZ}$ and $T^{Zh}$ have inferior performance than $T^S$ and $T_n$. This again shows that $T_n$ has the advantage a global smoothing test should have. This justifies that low-frequency models are in favor of global smoothing tests.  When $p=8$, $T^S$ is seriously  affected by the dimension, the influence of the dimension $p$ to $T_n$ is limited.

In the first two scenarios, $X\sim N(0,I_p)$  and thus the components of $X$ are uncorrelated from each other. Now we consider correlated case.

{\it Scenario 3. } Consider
\begin{eqnarray*}
  Y = \beta^{\top} X + a\cdot \exp(-(\beta^{\top} X)^2) + \varepsilon .
\end{eqnarray*}
Two settings are considered: 
$X\sim N(0,\Sigma)$ where $\Sigma = (0.5^{|i-j|})_{p\times p} $.
The error $\varepsilon\sim N(0,1)$.
We test with dimension $p = 2, 4, 8$, sample size $n=200$ and parameter $\beta = (1,1,\dots,1)^{\top} /\sqrt{p}$.
Results are presented in 
{\it Table \ref{tab: Scenario 3 dependent}}.

\begin{center}
Table \ref{tab: Scenario 3 dependent} about here.
\end{center}
When $p=2$ and $p=4$, $T^S$, $T^{GWZ}$ and our test have similar powers while $T^{Zh}$ does not work well.
When the dimension is raised up to $p=8$, $T^{GWZ}$ and $T_n$ become the winner.
As $T^{GWZ}$ adopts the dimension reduction structure under the null in this setting, its good performance is understandable.
$T_n$, however, requires no model structure information and performs similarly as $T^{GWZ}$.

The above three scenarios are all concerned with the linear model under the null hypothesis, therefore a nonlinear null model is used in the following scenario.
Denote $X_i$ as the $i$th component of $X$.

{\it Scenario 4. } Consider
\begin{eqnarray*}
  Y = \exp(c_1 X_1) + (c_2 X_2)^3 + c_3 \sin(\pi X_3) + c_4 |X_4| + c_5 X_5\cdot X_6
  + a \cdot \cos(\beta^{\top} X) + \varepsilon
\end{eqnarray*}
where $c_1 = c_2 = \dots = c_5 = 1/\sqrt 6 $ and $\beta = (0,1,1,0,\dots,0)^\top$.

This  model does not have a dimension reduction structure under the null hypothesis and thus it not in favor of $T^{GWZ}$ that is designed for  single index models.
 To make a comparison, we adopt its model adaptation idea by using the following test statistic as
\begin{eqnarray*}
  n h^{\hat q/2}\cdot \frac{1}{n(n-1)}\sum_{i=1}^{n}\sum_{j\neq i}
  \hat e_i \hat e_j \frac{1}{h^{\hat q}}K(\frac{\hat B_{\hat q}^\top(x_i-x_j)}{h}) .
\end{eqnarray*}

\begin{center}
  Figure \ref{fig: scenario 4} about here.
\end{center}

The results in Figure \ref{fig: scenario 4} clearly suggest that the proposed test $T_n$ performs much better than the competitors no matter they are either local or global smoothing tests. This further confirms the advantages of the new method.


\subsection{Real data analysis}

We now analyze the Auto MPG data set that can be 
downloaded from the UCI Machine Learning Repository \cite{Lichman:2013}.
\cite{quinlan1993combining} firstly used the data set and recently
\cite{xia2007constructive} and
\cite{guo2016model} as an illustration for their methods.
A linear regression model was build in 
\cite{quinlan1993combining}.
Here we use the proposed test to check the adequacy of the linear model.
The response variate $Y$ is mpg: miles per gallon.
The first 6 attributes, noted from $X_1$ to $X_6$, includes running year of the model, acceleration time from still state to 60 miles per hour, car weight, horsepower, displacement of the engine and the number of cylinders.
For the multi-valued independent variable origin, we introduce dummy variables as
\cite{guo2016model} and 
\cite{xia2007constructive} did.
One of the new indicator variables $X_7=1$ if the car is from America; otherwise, $X_7 = 0$.
Another dummy variable $X_8$ indicates whether the car is from Europe.
 The attributes are standardized one by one.
The $p$-value is about $0$ and
thus, the linear model is not suitable for this data set.
The result is coincident with 
\cite{guo2016model}.

\setcounter{equation}{0}
\section{Discussions}\label{sec: discussion}
In this paper, we build a bridge between local smoothing test and global smoothing test and propose a test that is local smoothing-based is of the global smoothing nature. Therefore, the test benefits both advantages of these two types of testing procedures. The theoretical properties and empirical studies confirm this nice feature. The approach may be applicable to other types of data and testing problems. These are ongoing.


\setcounter{equation}{0}
\section*{Acknowledgments}
The research described herewith was supported by a grant from The University Grants Council of Hong Kong.

\newpage
\setcounter{equation}{0}
\section*{Appendix}

\begin{proof}[Proofs of Lemma 2.2.]
Recall the formula of $\hat V$, the integral can be computed  as
\begin{eqnarray}\label{product of density}
  && \int_{R^p} \frac{1}{h}K\left(\frac{\alpha^{\top}(x_i-x_j)}{h}\right) \mu(\alpha) d \alpha \\
  &=& h^{-1} \int_{R^p}
  (2\pi)^{-1/2} \exp\left( -\frac{\alpha^{\top}(x_i-x_j)(x_i-x_j)^{\top} \alpha}{2 h^2} \right)
  \cdot (2\pi)^{-p/2} \sigma_\alpha^{-p} \exp \left( -\frac{\alpha^{\top} \alpha}{2\sigma^2} \right)
  d \alpha \nonumber\\
  &=& h^{-1} (2\pi)^{-1/2} \sigma_\alpha^{-p}
  \int_{R^p}
  (2\pi)^{-p/2}
   \exp\left( - \frac{1}{2}\alpha^{\top}
   \left(
  \frac{(x_i-x_j)(x_i-x_j)^{\top}}{h^2}+\frac{1}{\sigma^2}I_p
   \right)
   \alpha
   \right)
  d \alpha \nonumber .
\end{eqnarray}
Define $\Sigma_{ij}^{-1} = \frac{(x_i-x_j)(x_i-x_j)^{\top}}{h^2}+\frac{1}{ \sigma^2}I_p $. Then the integral \eqref{product of density} becomes
\begin{eqnarray}
  &&
  h^{-1} (2\pi)^{-1/2} \sigma^{-p}
  \int_{R^p}
  (2\pi)^{-p/2}
  \exp\left( - \frac{1}{2}\alpha^{\top}
  \Sigma_{ij}^{-1}
  \alpha
  \right)
  d \alpha \nonumber \\
  &=&
  h^{-1} (2\pi)^{-1/2} \sigma^{-p}
  |\Sigma_{ij}|^{1/2}
  \int_{R^p}
  (2\pi)^{-p/2}
  |\Sigma_{ij}|^{-1/2}
  \exp\left( - \frac{1}{2}\alpha^{\top}
  \Sigma_{ij}^{-1}
  \alpha
  \right)
  d \alpha \nonumber \\
  &=&
  h^{-1} (2\pi)^{-1/2} \sigma^{-p}
  |\Sigma_{ij}|^{1/2} \label{simplified with Sigma} .
\end{eqnarray}
Next we study the property of $\Sigma_{ij}$ to get $|\Sigma_{ij}|$.
Let
\begin{eqnarray}
  A_{ij} &=& \frac{(x_i-x_j)(x_i-x_j)^{\top}}{h^2} \nonumber .
\end{eqnarray}
The matrix $A_{ij}$ is symmetric and $rank(A_{ij}) = 1$.
By some algebraic calculations, $A_{ij}$ has a nonzero eigenvalue $\frac{d_{ij}}{h^2}$ 
where $d_{ij} = \|x_i-x_j \|^2 = (x_i-x_j)^{\top} (x_i-x_j)$.
Therefore, it is well known that $\Sigma_{ij}^{-1} = A_{ij}+I_p/\sigma^2$ can be decomposed as
\[
A\begin{pmatrix}
    \frac{d_{ij}}{h^2}+\frac{1}{\sigma^2} & & & \\
    & \frac{1}{\sigma^2} & & \\
    & & \ddots & \\
    & &  & \frac{1}{\sigma^2} .
\end{pmatrix} A^{\top},
\]
for some non-singular matrix $A$ where the elements on the diagonal of the  matrix inside are the eigenvalues of $\Sigma_{ij}^{-1}$ and the determination is
\begin{eqnarray}\label{determination of Sigma}
  |\Sigma_{ij}| = (\frac{d_{ij}}{h^2}+\frac{1}{\sigma^2} )^{-1} (\sigma^2)^{p-1}.
\end{eqnarray}
From  \eqref{product of density} -- \eqref{determination of Sigma}, the estimate in \eqref{hat V} can be written as
\begin{eqnarray}
  \hat V &=& \frac{1}{n(n-1)}\sum_{i=1}^n\sum_{j\neq i} \hat e_i\hat e_j
  h^{-1} (2\pi)^{-1/2} \sigma^{-p}
  (\frac{d_{ij}}{h^2}+\frac{1}{\sigma^2} )^{-1/2} \sigma^{p-1} \nonumber\\
  &=& \frac{1}{\sqrt{2\pi}} \cdot \frac{1}{n(n-1)}\sum_{i=1}^n\sum_{j\neq i} \hat e_i\hat e_j
  \frac{1}{\sqrt{\sigma^2 d_{ij}+h^2}}.
\end{eqnarray}
When $\sigma^2$ is chosen to be $h^2$, we have
\begin{eqnarray}\label{1/h times V_n}
   \hat V = \frac{1}{h\sqrt{2\pi}}\cdot \frac{1}{n(n-1)}\sum_{i=1}^n\sum_{j\neq i} \hat e_i\hat e_j \frac{1}{\sqrt{d_{ij}+1}}.
\end{eqnarray}
\end{proof}
\subsection*{Designed Conditions}
The following conditions are for  the consistency and asymptotic normality of $\hat\theta_n$. 
\begin{enumerate}[label=(\alph*)]
\item $\{(x_i,y_i)\}_{i=1,\dots,n} $ are i.i.d. random samples from $(X,Y) $ in $R^p\times R $ and $EY^2<\infty $.
\item The  parameter $\Theta$ is compact and convex.
\item The regression function $g(x,\theta)$ is a Borel measurable real function on $R^p$ for each $\theta$ and is twice continuously differentiable with respect to $\theta$ for each $x$.
\item Let  $\|\cdot\|$ represent the Euclidean norm.
\begin{eqnarray*}
  && E(\sup_{\theta\in \Theta} g^2(X,\theta))<\infty , \\
  && E(\sup_{\theta\in\Theta} \| \frac{\partial g(X,\theta)}{\partial \theta}\cdot \frac{\partial g(X,\theta)}{\partial \theta^{\top}} \|)<\infty,  \\
  && E(\sup_{\theta\in\Theta} \| (Y-g(X,\theta))^2 \cdot \frac{\partial g(X,\theta)}{\partial \theta}\cdot \frac{\partial g(X,\theta)}{\partial \theta^{\top}} \| < \infty, \\
  && E(\sup_{\theta\in\Theta}\| (Y-g(X,\theta)) \cdot \frac{\partial^2 g(X,\theta)}{\partial\theta\partial\theta^{\top}} \|)<\infty.
\end{eqnarray*}
\item There exists a unique minimizer $\theta^\ast$ such that
\begin{eqnarray*}
  \theta^\ast = \arg\inf_{\theta\in R^d} E(Y-g(X,\theta))^2 .
\end{eqnarray*}
Under the null hypothesis, $\theta^\ast$ is an interior point of $\Theta$.
\item The matrix $E(\frac{\partial g(X,\theta)}{\partial \theta}\cdot \frac{\partial g(X,\theta)}{\partial \theta^{\top}})$
is nonsingular.
\end{enumerate}


The following lemma shows the asymptotic property of $\hat \theta_n$.
\begin{lem}\label{asymptotic property of theta estimation}
Suppose that the above conditions are satisfied, we have the following asymptotic properties.
Denote $H_{\dot g} = E(\dot g(X,\theta^\ast) \dot g(X,\theta^\ast)^{\top}) $ .
\begin{enumerate}
  \item Under the null hypothesis, $\theta^\ast = \theta_0$ and
  \begin{eqnarray*}
  \sqrt n (\hat\theta_n - \theta^\ast) = H_{\dot g}^{-1}\frac{1}{\sqrt n}\sum_{i=1}^n \varepsilon_i \dot g(x_i,\theta^\ast)+o_p(1).
\end{eqnarray*}
  \item Under the local alternative models with $\delta_n\to 0$, $\theta^\ast = \theta_0$ and
  \begin{eqnarray*}
  \sqrt n (\hat\theta_n - \theta^\ast) = H_{\dot g}^{-1}\frac{1}{\sqrt n}\sum_{i=1}^n \varepsilon_i \dot g(x_i,\theta^\ast) + \sqrt n \delta_n \cdot H_{\dot g}^{-1}E(\ell \dot g)+o_p(1) .
\end{eqnarray*}
  \item Under the global alternative model with a fixed $\delta_n$, $\theta^\ast = \theta_1$ where
  \begin{eqnarray*}
    \theta_1 = \arg\min\limits_{\theta\in\Theta} E(g(X,\theta_0)-g(X,\theta)+\delta_n \ell(X))^2
  \end{eqnarray*}
  and
  \begin{eqnarray*}
    \hat\theta_n-\theta^\ast &=& O_p(\frac{1}{\sqrt n}).
  \end{eqnarray*}
\end{enumerate}

\end{lem}

\begin{proof}[Proofs of Lemma \ref{asymptotic property of theta estimation}.]
The least squares estimate of $\theta_0$ is the minimizer of   the following function over all $\theta \in \Theta$ as
\begin{eqnarray*}
  Q(\theta) &=& \frac{1}{n}\sum_{i=1}^{n}(y_i-g(x_i,\theta))^2, \\
  \hat\theta_n &=& \arg\min\limits_{\theta\in\Theta} Q(\theta) .
\end{eqnarray*}
The first order derivative of $Q$ with respect to $\theta$ is
\begin{eqnarray*}
  \dot Q(\theta) &=& -\frac{2}{n}\sum_{i=1}^{n}(y_i-g(x_i,\theta))\dot g(x_i,\theta)
\end{eqnarray*}
where $\dot g = \partial g/\partial\theta$.
The second order derivative of $Q$ with respect to $\theta$ is
\begin{eqnarray*}
  \ddot Q(\theta) &=& \frac{2}{n}\sum_{i=1}^{n}\dot g(x_i,\theta)g(x_i,\theta)^{\top}-\frac{2}{n}\sum_{i=1}^{n}(y_i-g(x_i,\theta))\ddot g(x_i,\theta) \\
\end{eqnarray*}
The least squares estimate $\hat\theta_n$ satisfies $\dot Q(\hat\theta_n)=0$.
Notice that $\frac{1}{n}\sum_{i=1}^{n}(y_i-g(x_i,\theta))^2 \overset{a.s.}{\to} E(Y-g(X,\theta))^2$ for all $\theta$, the estimator $\hat\theta_n\overset{a.s.}{\to}\theta^\ast$.
Applying the Taylor expansion to $\dot Q(\hat\theta_n)$ around $\theta^\ast$, we have
\begin{eqnarray*}
  \dot Q(\hat\theta_n)-\dot Q(\theta^\ast) &=& \ddot Q(\tilde\theta)(\hat\theta_n-\theta^\ast) \\
  \hat\theta_n-\theta^\ast &=& \ddot Q(\tilde\theta)^{-1}(\dot Q(\hat\theta_n)-\dot Q(\theta^\ast)) \\
  &=& -\ddot Q(\tilde\theta)^{-1}\dot Q(\theta^\ast)
\end{eqnarray*}
where $\tilde\theta$ is a mid-value between $\hat\theta_n$ and $\theta^\ast$.
Since $\tilde\theta$ is close to $\theta^\ast$, it is easy to  show that
\begin{eqnarray*}
  \sqrt n (\hat\theta_n-\theta^\ast) &=& - E(\ddot Q(\theta^\ast))^{-1} \sqrt n \dot Q(\theta^\ast)+ o_p(1) \\
  &=& 2 E(\ddot Q(\theta^\ast))^{-1}\frac{1}{\sqrt n}\sum_i (y_i-g(x_i,\theta^\ast))\dot g(x_i,\theta^\ast) + o_p(1).
\end{eqnarray*}
Under the null hypothesis and the local alternatives with $\delta_n\to 0$,
\begin{eqnarray*}
    \inf_{\theta\in\Theta} E(Y-g(X,\theta))^2 &=& \inf_{\theta\in\Theta} E(g(X,\theta_0)-g(X,\theta))^2+E(\varepsilon)^2 .
\end{eqnarray*}
thus $\theta^\ast = \theta_0$.
Specifically, under the null hypothesis, $y_i-g(x_i,\theta^\ast) = \varepsilon_i$ and
\begin{eqnarray*}
  E(\ddot Q(\theta^\ast)) &=& 2 E(\dot g(X,\theta^\ast)\dot g(X,\theta^\ast)^\top) + 2 E(\varepsilon\ddot g)\\
  &=& 2 H_{\dot g}
\end{eqnarray*}
thus
\begin{eqnarray*}
  \sqrt n(\hat\theta_n-\theta^\ast) &=& H_{\dot g}^{-1}\frac{1}{\sqrt n}\sum_{i=1}^n \varepsilon_i \dot g(x_i,\theta^\ast)+o_p(1).
\end{eqnarray*}

Under the local alternative hypothesis, $y_i-g(x_i,\theta^\ast) = \varepsilon_i + \delta_n\ell(x_i)$ and
\begin{eqnarray*}
  E(\ddot Q(\theta^\ast)) &=& 2 E(\dot g(X,\theta^\ast)\dot g(X,\theta^\ast)^\top) + 2 E(\varepsilon\ddot g)+ 2 \delta_n E(\ell\dot g) \\
  &=& 2 H_{\dot g}+o_p(1).
\end{eqnarray*}
Hence
\begin{eqnarray*}
  \sqrt n(\hat\theta_n-\theta^\ast) &=& H_{\dot g}^{-1}\frac{1}{\sqrt n}\sum_{i=1}^n \varepsilon_i \dot g(x_i,\theta^\ast)
  + \sqrt n \delta_n\cdot H_{\dot g}^{-1}\frac{1}{n}\sum_{i=1}^n \ell(x_i) \dot g(x_i,\theta^\ast)
  + o_p(1) \\
  &=& H_{\dot g}^{-1}\frac{1}{\sqrt n}\sum_{i=1}^n \varepsilon_i \dot g(x_i,\theta^\ast)
  + \sqrt n \delta_n\cdot H_{\dot g}^{-1}E(\ell\dot g)
  + o_p(1) .
\end{eqnarray*}
Under the global alternative with $\delta_n$ fixed,
\begin{eqnarray*}
    \inf_{\theta\in\Theta} E(Y-g(X,\theta))^2 &=& \inf_{\theta\in\Theta} E(g(X,\theta_0)-g(X,\theta)+\delta_n\ell(X))^2+E(\varepsilon)^2 .
\end{eqnarray*}
The minimizer
\begin{eqnarray*}
  \theta^\ast = \theta_1 =\arg \inf_{\theta\in\Theta} E(Y-g(X,\theta))^2
\end{eqnarray*}
is a value that is more likely to be different from $\theta_0$ under the null hypothesis.
In this case,
\begin{eqnarray*}
  E(\ddot Q(\theta^\ast)) &=& 2 H_{\dot g}-2 E\{[g(X,\theta_0)-g(X,\theta_1)+\delta_n\ell(X)]\ddot g(X,\theta_1)\} \\
  \sqrt n (\hat\theta_n-\theta^\ast)
  &=& 2 E(\ddot Q(\theta^\ast))^{-1}\frac{1}{\sqrt n}\sum_i (y_i-g(x_i,\theta^\ast))\dot g(x_i,\theta^\ast) + o_p(1).
\end{eqnarray*}
Notice that at the population level,
\begin{eqnarray*}
  0 = \frac{\partial E(Y-g(X,\theta))^2}{\partial \theta}|_{\theta=\theta^\ast}
  = -2 E[(Y-g(X,\theta^\ast))\dot g(X,\theta^\ast)] .
\end{eqnarray*}
Therefore under the global alternative hypothesis,
\begin{eqnarray*}
  \sqrt n (\hat\theta_n-\theta^\ast)
  &\overset{d}{\to}& N(0,\Sigma_1)
\end{eqnarray*}
where $\Sigma_1 = 4 E(\ddot Q(\theta^\ast))^{-1} E((Y-g(X,\theta^\ast))^2 \dot g(X,\theta^\ast)\dot g(X,\theta^\ast)^\top) E(\ddot Q(\theta^\ast))^{-1}  $.
\end{proof}

\subsection*{Proofs of the theorems}
Define $w_{ij} = \frac{1}{\sqrt{\lambda_{ij}+1}} $ which is symmetric about $x_i$ and $x_j$. The integrated statistic $V_n$ can be decomposed as:
\begin{eqnarray}\label{V_n}
  V_n &=& \frac{1}{n(n-1)}\sum_{i=1}^{n}\sum_{j\neq i} \hat e_i \hat e_j w_{ij} \nonumber\\
  &=& \frac{1}{n(n-1)}\sum_{i=1}^{n}\sum_{j\neq i} e_i e_j w_{ij} - \frac{2}{n(n-1)}\sum_{i=1}^{n}\sum_{i\neq j} e_i (e_j-\hat e_j) w_{ij} \nonumber \\
  && + \frac{1}{n(n-1)}\sum_{i=1}^{n}\sum_{j\neq i} (e_i-\hat e_i)(e_j- \hat e_j) w_{ij} \nonumber\\
  &=:& V_1-2V_2+V_3.
\end{eqnarray}

\begin{proof}[Proofs of Theorem \ref{theorem: null}. ]
Under the null hypothesis, $e = \varepsilon $,
\begin{eqnarray*}
  V_1 &=& \frac{1}{n(n-1)}\sum_{i=1}^{n}\sum_{j\neq i} \varepsilon_i \varepsilon_j w_{ij} , \\
  V_2 &=& \frac{1}{n(n-1)}\sum_{i=1}^{n}\sum_{j\neq i} \varepsilon_i (\varepsilon_j-\hat \varepsilon_j)w_{ij} \\
  &=& \frac{1}{n(n-1)}\sum_{i=1}^{n}\sum_{j\neq i} \varepsilon_i \dot g_j^{\top} (\hat\theta_n-\theta_0) w_{ij} + O_p(n^{-3/2}) \\
  &=& \frac{1}{n}\sum_{i=1}^n \varepsilon_i (\frac{1}{n-1}\sum_{j\neq i} \dot g_j w_{ij})^{\top}(\hat\theta_n-\theta_0) + O_p(n^{-3/2})  \\
  &=& \frac{1}{n}\sum_{i=1}^n \varepsilon_i (\hat\theta_n-\theta_0)^{\top} E_{1i}  + O_p(n^{-3/2})  \\
  &=& \frac{1}{n^2}\sum_{i=1}^n \sum_{j=1}^n \varepsilon_i \varepsilon_j  \dot g_j^{\top} H_{\dot g}^{-1} E_{1i} + O_p(n^{-3/2}) \\
  &=& \frac{n-1}{n}\cdot \frac{1}{n(n-1)} \sum_{i=1}^n \sum_{j\neq i}^n \varepsilon_i \varepsilon_j  \dot g_j^{\top} H_{\dot g}^{-1} E_{1i}
  + \frac{1}{n^2} \sum_{i=1}^n \varepsilon_i^2 \dot g_i^{\top} H_{\dot g}^{-1} E_{1i} + O_p(n^{-3/2}) .
  \end{eqnarray*}
Define
\begin{eqnarray*}
  V_2^0 &=& \frac{1}{n(n-1)}\sum_{i=1}^n \sum_{j\neq i} \varepsilon_i \varepsilon_j  \dot g_j^{\top} H_{\dot g}^{-1} E_{1i} =  O_p(\frac{1}{n}), ~~~E_{1i} = E_j(\dot g_j w_{ij}) , \\
  \mu_{v_2}
  &=& E(\varepsilon_i^2 \dot g_i^{\top} H_{\dot g}^{-1} E_{1i} )
  = E(\sigma_i^2 \dot g_i^{\top} H_{\dot g}^{-1} E_{1i} ), ~~~\sigma_i^2 = E(\varepsilon^2|x_i).
\end{eqnarray*}
Then
\begin{eqnarray}\label{nV_2}
  n V_2 = n V_2^0 + \mu_{v_2} + O_p(n^{-1/2})
\end{eqnarray}
For $V_3$, we have a similar decomposition as,
\begin{eqnarray*}
  V_3
  &=& \frac{1}{n(n-1)}\sum_{i=1}^{n}\sum_{j\neq i} (\varepsilon_i-\hat \varepsilon_i)(\varepsilon_j- \hat \varepsilon_j) w_{ij} \\
  &=& \frac{1}{n(n-1)}\sum_{i=1}^{n}\sum_{j\neq i} (\hat\theta_n-\theta_0)^{\top} \dot g_i \dot g_j ^{\top} (\hat\theta_n-\theta_0) w_{ij} +O_p(n^{-3/2}) \\
  &=& (\hat\theta_n-\theta_0)^{\top} \cdot \frac{1}{n(n-1)}\sum_{i=1}^{n}\sum_{j\neq i} \dot g_i \dot g_j ^{\top} w_{ij}\cdot  (\hat\theta_n-\theta_0) +O_p(n^{-3/2}) \\
  &=& \frac{1}{n^2}\sum_{i=1}^n \sum_{j=1}^n \varepsilon_i \varepsilon_j \dot g_i^{\top} H_{\dot g}^{-1}  E(\dot g_k E_{1k}^{\top}) H_{\dot g}^{-1} \dot g_j
  +O_p(n^{-3/2})  \\
  &=& \frac{1}{n^2}\sum_{i=1}^n \sum_{j\neq i}^n \varepsilon_i \varepsilon_j \dot g_i^{\top} H_{\dot g}^{-1}  E(\dot g_k E_{1k}^{\top}) H_{\dot g}^{-1} \dot g_j
  + \frac{1}{n^2}\sum_{i=1}^n \varepsilon_i^2 \dot g_i^{\top} H_{\dot g}^{-1}  E(\dot g_k E_{1k}^{\top}) H_{\dot g}^{-1} \dot g_i +O_p(n^{-3/2}) .
  \end{eqnarray*}
Define
\begin{eqnarray*}
  V_3^0 &=& \frac{1}{n(n-1)}\sum_{i=1}^n \sum_{j\neq i} \varepsilon_i \varepsilon_j \dot g_i^{\top} H_{\dot g}^{-1}  E(\dot g_k E_{1k}^{\top}) H_{\dot g}^{-1} \dot g_j =
  O_p(\frac{1}{n}) , \\
  \mu_{v_3} &=& E(\varepsilon^2_i\dot g_i^{\top} H_{\dot g}^{-1} E(\dot g_k E_{1k}^{\top}) H_{\dot g}^{-1} \dot g_i )
  = E(\sigma^2_i\dot g_k^{\top} H_{\dot g}^{-1}\dot g_i\dot g_i^{\top} H_{\dot g}^{-1} E_{1k} ) .
\end{eqnarray*}
Then
\begin{eqnarray}\label{nV_3}
  n V_3 = n V_3^0 +\mu_{v_3} +O_p(n^{-1/2}) .
\end{eqnarray}

If $\sigma_\varepsilon^2 \equiv E(\varepsilon^2|X)$, then
\begin{eqnarray*}
    \mu_{v_3} = \sigma_\varepsilon^2 E(\dot g_i^{\top} H_{\dot g}^{-1} E_{1i})= \mu_{v_2} .
\end{eqnarray*}
Then from \eqref{V_n} to \eqref{nV_3} we can see $ n V_n $ has the same asymptotic behavior as $ n V_n^0 +\mu_0 $, i.e.
\begin{eqnarray}\label{nV_n in null}
  n V_n = n V_n^0 + \mu_0 + o_p(1)
\end{eqnarray}
where $V_n^0 = V_1-2 V_2^0 + V_3^0 $ and $\mu_0 = \mu_{v_3}-2\mu_{v_2}$.

Denote
\begin{eqnarray*}
  \tilde w_{ij} = w_{ij}
- 2 \dot g_j^{\top} H_{\dot g}^{-1} E_{1i}
+ \dot g_i^{\top} H_{\dot g}^{-1}  E(\dot g_k E_{1k}^{\top}) H_{\dot g}^{-1} \dot g_j,
\end{eqnarray*}
then
\begin{eqnarray*}
  V_n^0 = \frac{1}{n(n-1)}\sum_{i=1}^n \sum_{j\neq i} \varepsilon_i \varepsilon_j \tilde w_{ij} .
\end{eqnarray*}

$V_n^0$ can be represented by a U-statistic.
Let
\begin{eqnarray*}
  h(x_i,x_j) = \frac{1}{2}\left(\tilde w_{ij}+\tilde w_{ji}\right),
\end{eqnarray*}
then $V_n^0$ has the same limiting distribution as
\begin{eqnarray}\label{U_n}
  U_n = \frac{1}{n(n-1)}\sum_{i=1}^n \sum_{j\neq i} \varepsilon_i \varepsilon_j h(x_i,x_j).
\end{eqnarray}
This is a degenerate U-statistic and details of its asymptotic distribution can be found in chapter 5, Serfling (1980). 
Here we simply show the results.
Let $\lambda_1,\lambda_2,\dots$ be the corresponding eigenvalues to the distinct solutions $s_1,s_2,\dots$ satisfying that
\begin{eqnarray*}
  \int_{-\infty}^\infty h(x_i,x_j)s(x_j)dF(x_j) = \lambda s(x_i)
\end{eqnarray*}
where $F(\cdot)$ is the cumulative distribution function of $X$.
Based on \eqref{nV_n in null} and \eqref{U_n}, the limiting null distribution of our test statistic is
  \begin{eqnarray*}
    T_n \overset{\Delta}{=} n V_n \overset{d}{\to} \sum_{i=1}^\infty \lambda_i(Z_i^2-1) + \mu_0
  \end{eqnarray*}
  where $Z_1,Z_2,\dots$ are independent standard normal random variables.
\end{proof}

\begin{proof}[Proofs of Theorem \ref{theorem: alternative}. ]
Under the alternative hypothesis,
\begin{eqnarray*}
  Y = g(X,\theta)+\delta_n \ell(X)+\varepsilon .
\end{eqnarray*}
therefore $e_i = \delta_n\ell(x_i)+\varepsilon_i $.

Firstly, consider the global alternative hypothesis where $\delta_n$ is some constant.
$V_n$ can be decomposed as
\begin{eqnarray}\label{V_n in alter}
  V_n &=& \frac{1}{n(n-1)}\sum_{i=1}^{n}\sum_{j\neq i} \hat e_i \hat e_j w_{ij} \nonumber \\
  &=& \frac{1}{n(n-1)}\sum_{i=1}^{n}\sum_{j\neq i} e_i e_j w_{ij} - \frac{2}{n(n-1)}\sum_{i=1}^{n}\sum_{i\neq j} e_i (e_j-\hat e_j) w_{ij} \nonumber \\
  && + \frac{1}{n(n-1)}\sum_{i=1}^{n}\sum_{j\neq i} (e_i-\hat e_i)(e_j- \hat e_j) w_{ij} \nonumber \\
  &=& V_1-2V_2+V_3 .
\end{eqnarray}
For the second term,
\begin{eqnarray}\label{V_2 in alter}
  V_2 &=& \frac{1}{n(n-1)}\sum_{i=1}^{n}\sum_{j\neq i} e_i
  (e_j-\hat e_j)
  w_{ij} \nonumber \\
  &=& \frac{1}{n(n-1)}\sum_{i=1}^{n}\sum_{j\neq i} e_i
  \dot g_j^\top (\hat\theta_n-\theta_0)
  w_{ij} + o_p(V_2^\ast) ,  \\
  V_2^\ast &=& \frac{1}{n(n-1)}\sum_{i=1}^{n}\sum_{j\neq i} e_i
  \dot g_j^\top (\hat\theta_n-\theta_0)
  w_{ij}\nonumber \\
  &=& \frac{1}{n(n-1)}\sum_{i=1}^{n}\sum_{j\neq i} e_i
  \dot g_j^\top
  w_{ij}
  \cdot o_p(1) \nonumber\\
  &=& E(\ell_i
  \dot g_j^\top
  w_{ij}) \cdot o_p(1) \nonumber\\
  &=& o_p(1) .
\end{eqnarray}
For the third term,
\begin{eqnarray}
  V_3 &=& \frac{1}{n(n-1)}\sum_{i=1}^{n}\sum_{j\neq i} (e_i-\hat e_i)(e_j- \hat e_j) w_{ij} \nonumber\\
  &=& (\hat\theta_n-\theta_0)^{\top} \cdot \frac{1}{n(n-1)}\sum_{i=1}^{n}\sum_{j\neq i} \dot g_i \dot g_j ^{\top} w_{ij}\cdot  (\hat\theta_n-\theta_0) + o_p(V_3^\ast) ,\label{V_3 in alter}\\
  V_3^\ast &=& (\hat\theta_n-\theta_0)^{\top} \cdot \frac{1}{n(n-1)}\sum_{i=1}^{n}\sum_{j\neq i} \dot g_i \dot g_j ^{\top} w_{ij}\cdot  (\hat\theta_n-\theta_0) \nonumber\\
  &=& o_p(1) \cdot E(\dot g_i \dot g_j^\top w_{ij})\cdot o_p(1) + o_p(1)\nonumber \\
  &=& o_p(1) \label{V_3* in alter}.
\end{eqnarray}
Hence from \eqref{V_n in alter} to \eqref{V_3* in alter}, we have $V_n = V_1+o_p(1)$ . When $\delta_n$ is fixed,
\begin{eqnarray*}
  V_1 &\overset{p}{\to}& \mu_1 = E(e_i e_j w_{ij}) = E(\delta_n^2 \ell(x_i)\ell(x_j)w(x_i,x_j)) .
\end{eqnarray*}
Therefore under the global alternative hypothesis,
  \begin{eqnarray*}
  T_n/n\overset{p}{\to}\mu_1 .
\end{eqnarray*}

Next, we consider the local alternative where $\delta_n\to 0$.
Similar with the proof under the null distribution, we have
\begin{eqnarray*}
  V_n &=& \frac{1}{n(n-1)}\sum_{i=1}^n\sum_{j\neq i}\hat e_i\hat e_j w_{ij} \\
  &=& \frac{1}{n(n-1)}\sum_{i=1}^n\sum_{j\neq i} e_i e_j \tilde w_{ij}
  + \frac{1}{n} \mu_0 + o_p(\frac{1}{n}) \\
  &=& \frac{1}{n(n-1)}\sum_{i=1}^n\sum_{j\neq i} \varepsilon_i \varepsilon_j \tilde w_{ij}
  + \delta_n \cdot \frac{1}{n(n-1)}\sum_{i=1}^n\sum_{j\neq i} [\ell_i \varepsilon_j+ \varepsilon_i \ell_j ] \tilde w_{ij} \\
  && + \delta_n^2\cdot \frac{1}{n(n-1)}\sum_{i=1}^n\sum_{j\neq i}  \ell_i\ell_j \tilde w_{ij}
  + \frac{1}{n} \mu_0+ o_p(\frac{1}{n}) \\
  &=& \frac{1}{n(n-1)}\sum_{i=1}^n\sum_{j\neq i} \varepsilon_i \varepsilon_j \tilde w_{ij}
  + \delta_n \cdot \frac{1}{n(n-1)}\sum_{i=1}^n\sum_{j\neq i} \varepsilon_i \ell_j (\tilde w_{ij}+\tilde w_{ji}) \\
  && + \delta_n^2\cdot \frac{1}{n(n-1)}\sum_{i=1}^n\sum_{j\neq i}  \ell_i\ell_j \tilde w_{ij} + \frac{1}{n} \mu_0  + o_p(\frac{1}{n})\\
  &=& \frac{1}{n(n-1)}\sum_{i=1}^n\sum_{j\neq i} \varepsilon_i
  [\varepsilon_j \tilde w_{ij} + \delta_n \ell_j (\tilde w_{ij}+\tilde w_{ji})] \\
  && + \delta_n^2\cdot \frac{1}{n(n-1)}\sum_{i=1}^n\sum_{j\neq i}  \ell_i\ell_j \tilde w_{ij} + \frac{1}{n} \mu_0+ o_p(\frac{1}{n}) .
\end{eqnarray*}

Define $M_i = E\left(\ell_j(\tilde w_{ij}+\tilde w_{ji})|x_i\right) $ and $\Delta_\mu = E(\ell_i\ell_j \tilde w_{ij})$.
We have
\begin{eqnarray*}
  && E[\varepsilon_j \tilde w_{ij} + \delta_n \ell_j (\tilde w_{ij}+\tilde w_{ji})|x_i]= \delta_n E[\ell_j(\tilde w_{ij}+\tilde w_{ji})|x_i]
  = \delta_n M_i
\end{eqnarray*}
and our test statistic
\begin{eqnarray}\label{T_n in local}
  T_n &=& n V_n
  = \sqrt{n} \delta_n\cdot \frac{1}{\sqrt n} \sum_{i=1}^n \varepsilon_i M_i + n \delta_n^2\cdot \Delta_\mu + \mu_0 +o_p(1) .
\end{eqnarray}
From the expression in \eqref{T_n in local}, the asymptotic behavior of $T_n$ can obtained.
when $\delta_n = n^{-1/2}$, $T_n \overset{d}{\to} N(\mu_{1n},\Sigma) $ where $\mu_{1n} = \Delta_\mu+\mu_0 $ and $\Sigma = Var(\varepsilon_i M_i)$.
When $\delta_n\to 0$ and $\sqrt n \delta_n\to \infty $ , $T_n/(n \delta_n^2)\overset{p}{\to}\Delta_\mu $.
\end{proof}

\bibliography{ref_local_global_p}
\bibliographystyle{dcu} 

\newpage
\begin{table}[htbp]
  \centering
  \caption{Empirical sizes and powers of $T_n$, $T^{GWZ}$, $T^S$ and $T^{Zh}$ for {\it Scenario 1} with $X\sim N(0,I_p)$, $\varepsilon\sim N(0,1) $ and $n=200$. }
    \begin{tabular}{r|cccc}
    \toprule
    \toprule
    p=2   & $T_n$ & $T^{GWZ}$  & $T^S$  & $T^{Zh}$ \\
    \midrule
    a = 0.0   & 0.0570 & 0.054 & 0.066 & 0.043 \\
    0.2   & 0.3790 & 0.403 & 0.372 & 0.211 \\
    0.6   & 0.9950 & 1.000 & 0.998 & 0.991 \\
    1.0   & 0.9990 & 1.000 & 1.000 & 1.000 \\
    \midrule
    \midrule
    p=4    & $T_n$ & $T^{GWZ}$  & $T^S$  & $T^{Zh}$ \\
    \midrule
    a = 0.0   & 0.0480 & 0.050 & 0.044 & 0.042 \\
    0.2   & 0.3510 & 0.384 & 0.200 & 0.100 \\
    0.6   & 0.9620 & 0.999 & 0.936 & 0.684 \\
    1.0   & 0.9860 & 1.000 & 1.000 & 0.997 \\
    \midrule
    \midrule
    p=8    & $T_n$ & $T^{GWZ}$  & $T^S$  & $T^{Zh}$ \\
    \midrule
    a = 0.0   & 0.0610 & 0.050 & 0.028 & 0.036 \\
    0.2   & 0.3400 & 0.372 & 0.052 & 0.047 \\
    0.6   & 0.9290 & 1.000 & 0.104 & 0.371 \\
    1.0   & 0.9650 & 1.000 & 0.202 & 0.847 \\
    \bottomrule
    \bottomrule
    \end{tabular}%
  \label{tab: Scenario 1}%
\end{table}%

\begin{figure}
  \centering
  \includegraphics[width = 12cm]{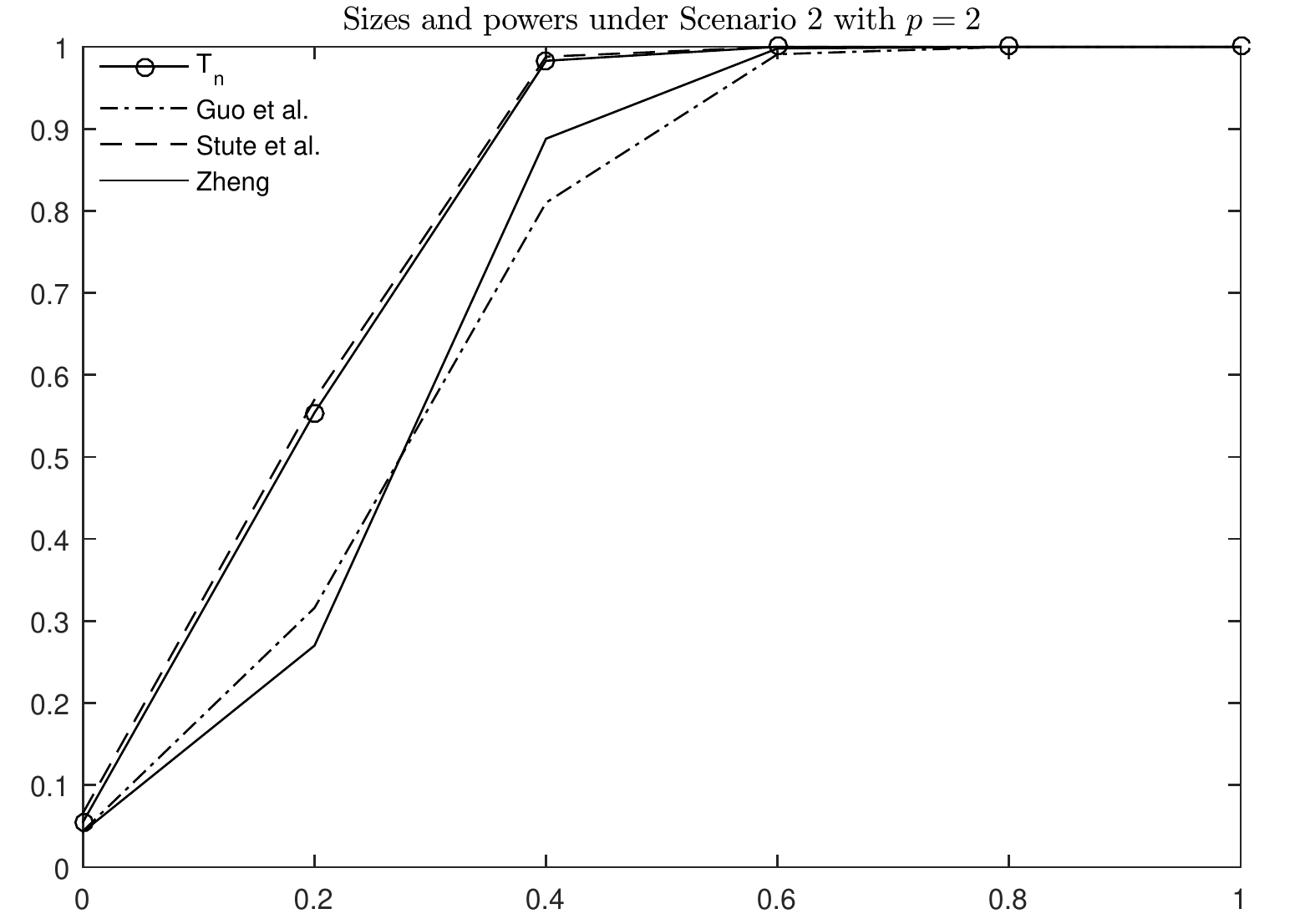}
  \vskip 2pt
  \includegraphics[width = 12cm]{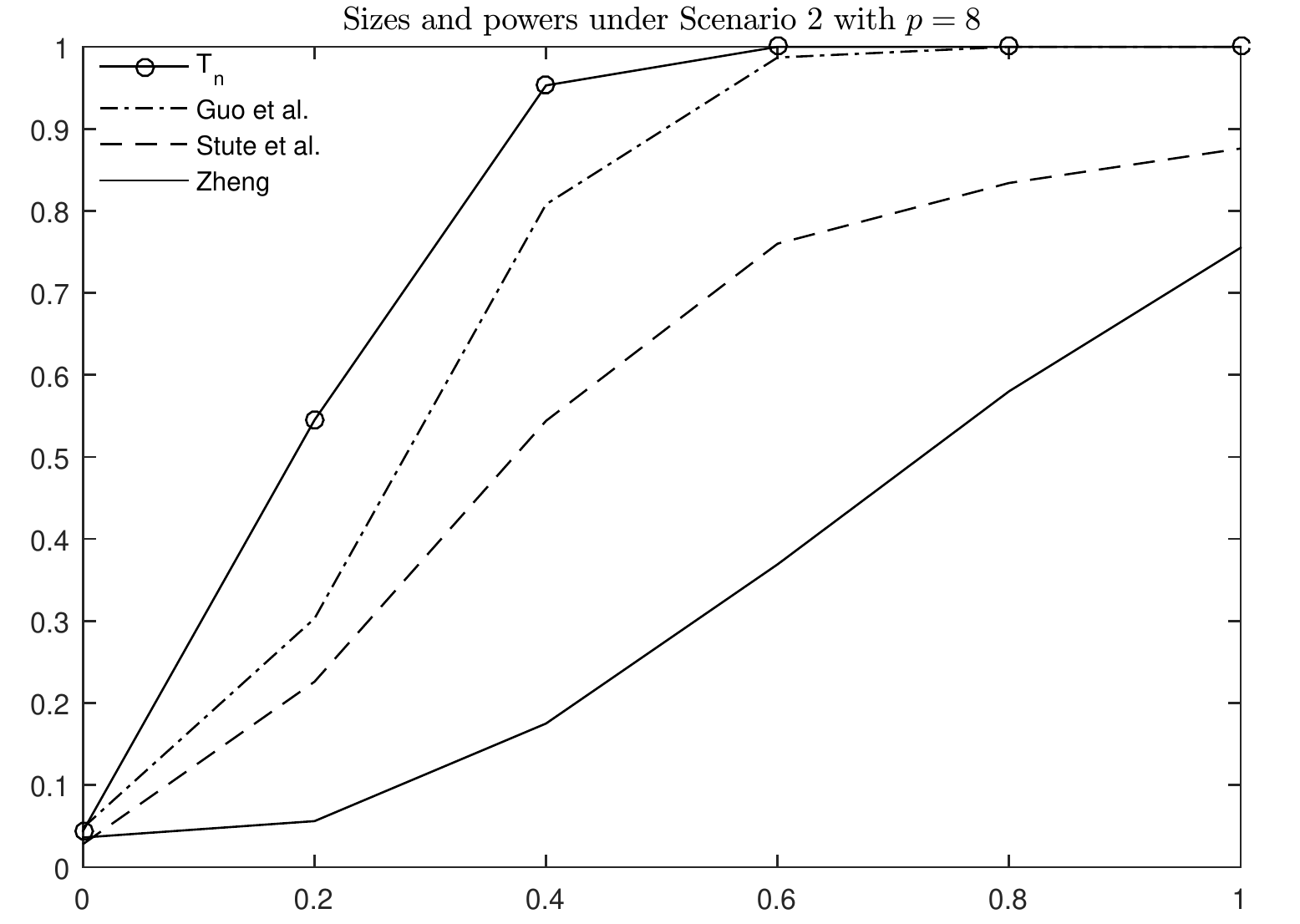}
  \caption{The empirical size and powers curves of $T_n$, $T^{GWZ}$, $T^S$ and $T^{Zh}$ in {\it Scenario 2} $n=200$ with $p=2$ and $p=8$. }
  \label{fig: scenario 2}
\end{figure}

\begin{table}[htbp]
  \centering
  \caption{Empirical sizes and powers of $T_n$, $T^{GWZ}$, $T^S$ and $T^{Zh}$ for {\it Scenario 3} with $X\sim N(0,\Sigma)$ and $n=200$. }
    \begin{tabular}{r|cccc}
    \toprule
    \toprule
    p=2   & $T_n$ & $T^{GWZ}$  & $T^S$  & $T^{Zh}$ \\
    \midrule
    a = 0.0   & 0.056 & 0.046 & 0.042 & 0.0465 \\
    0.2   & 0.339 & 0.274 & 0.32  & 0.1675 \\
    0.6   & 0.995 & 0.995 & 0.996 & 0.959 \\
    1.0   & 1     & 1.000 & 1     & 1 \\
    \midrule
    \midrule
    p=4   & $T_n$ & $T^{GWZ}$  & $T^S$  & $T^{Zh}$ \\
    \midrule
    a = 0.0   & 0.061 & 0.0465 & 0.058 & 0.047 \\
    0.2   & 0.266 & 0.213 & 0.192 & 0.0805 \\
    0.6   & 0.982 & 0.9795 & 0.92  & 0.5385 \\
    1.0   & 1     & 1     & 1     & 0.966 \\
    \midrule
    \midrule
    p=8   & $T_n$ & $T^{GWZ}$  & $T^S$  & $T^{Zh}$ \\
    \midrule
    a = 0.0   & 0.044 & 0.052 & 0.046 & 0.0365 \\
    0.2   & 0.243 & 0.211 & 0.1   & 0.054 \\
    0.6   & 0.944 & 0.9595 & 0.482 & 0.2875 \\
    1.0   & 1     & 1     & 0.85  & 0.743 \\
    \bottomrule
    \bottomrule
    \end{tabular}%
  \label{tab: Scenario 3 dependent}%
\end{table}%

\begin{figure}
 \centering
 \includegraphics[width = 12cm, height = 6cm]{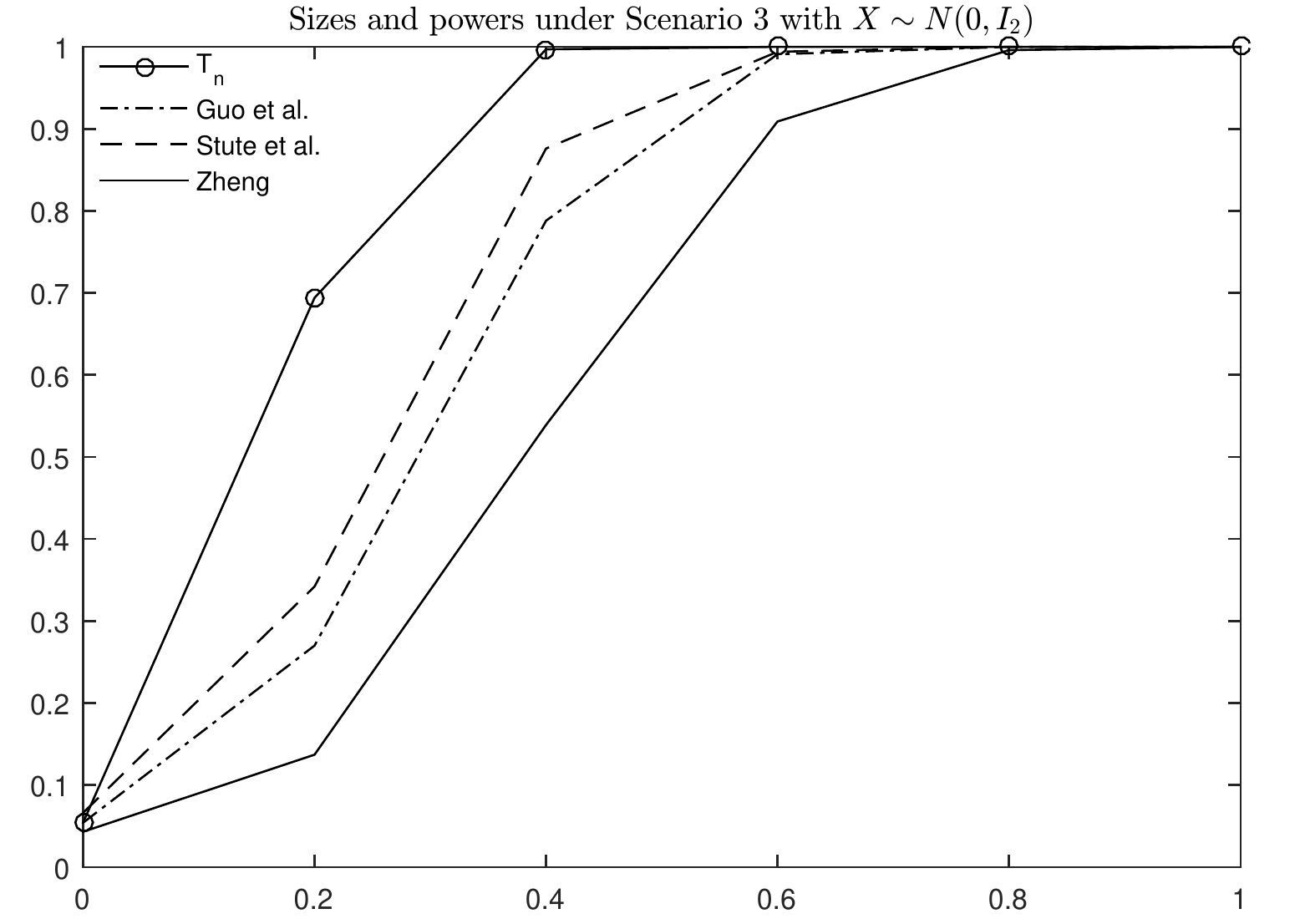}
 \includegraphics[width = 12cm, height = 6cm]{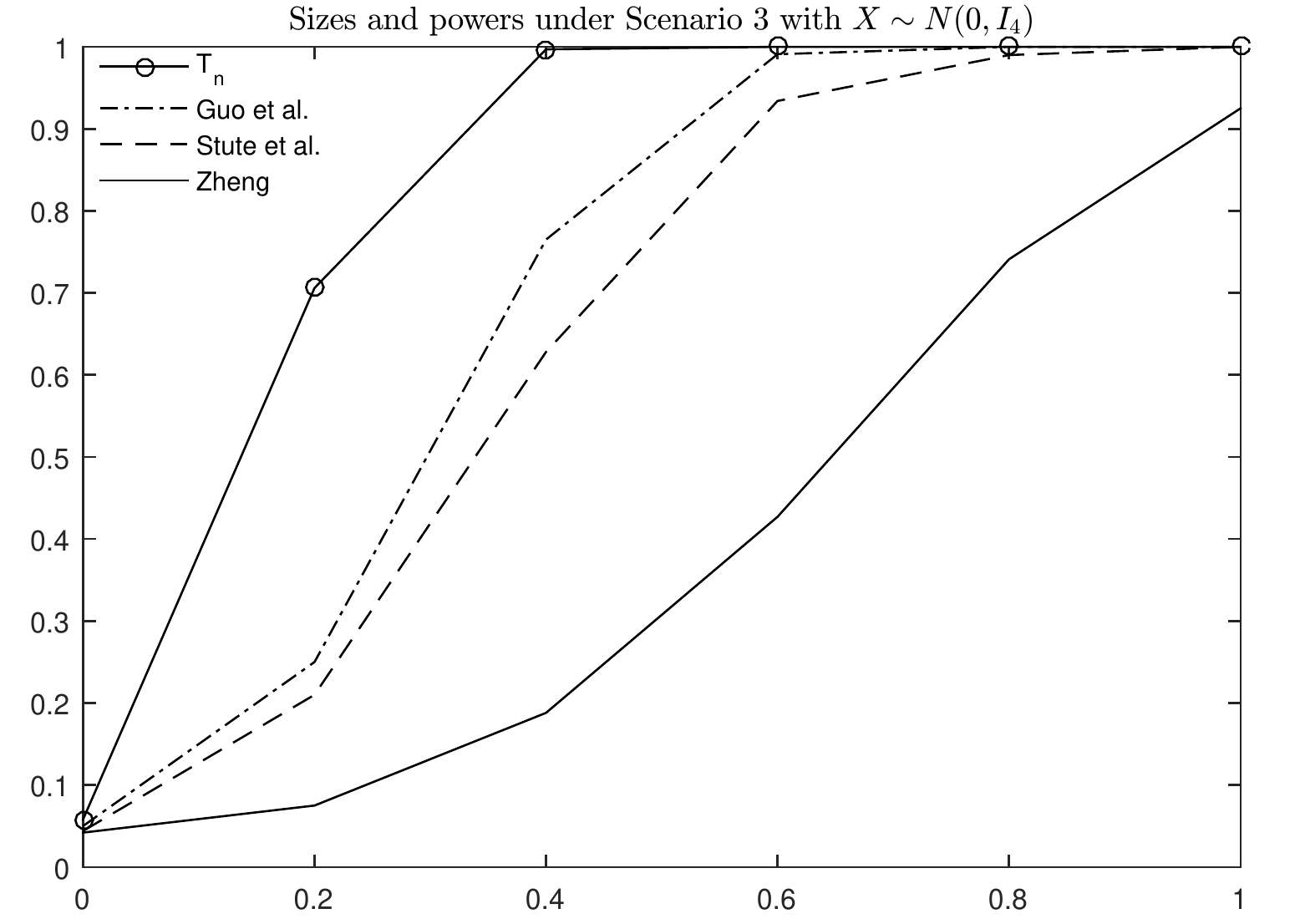}
 \includegraphics[width = 12cm, height = 6cm]{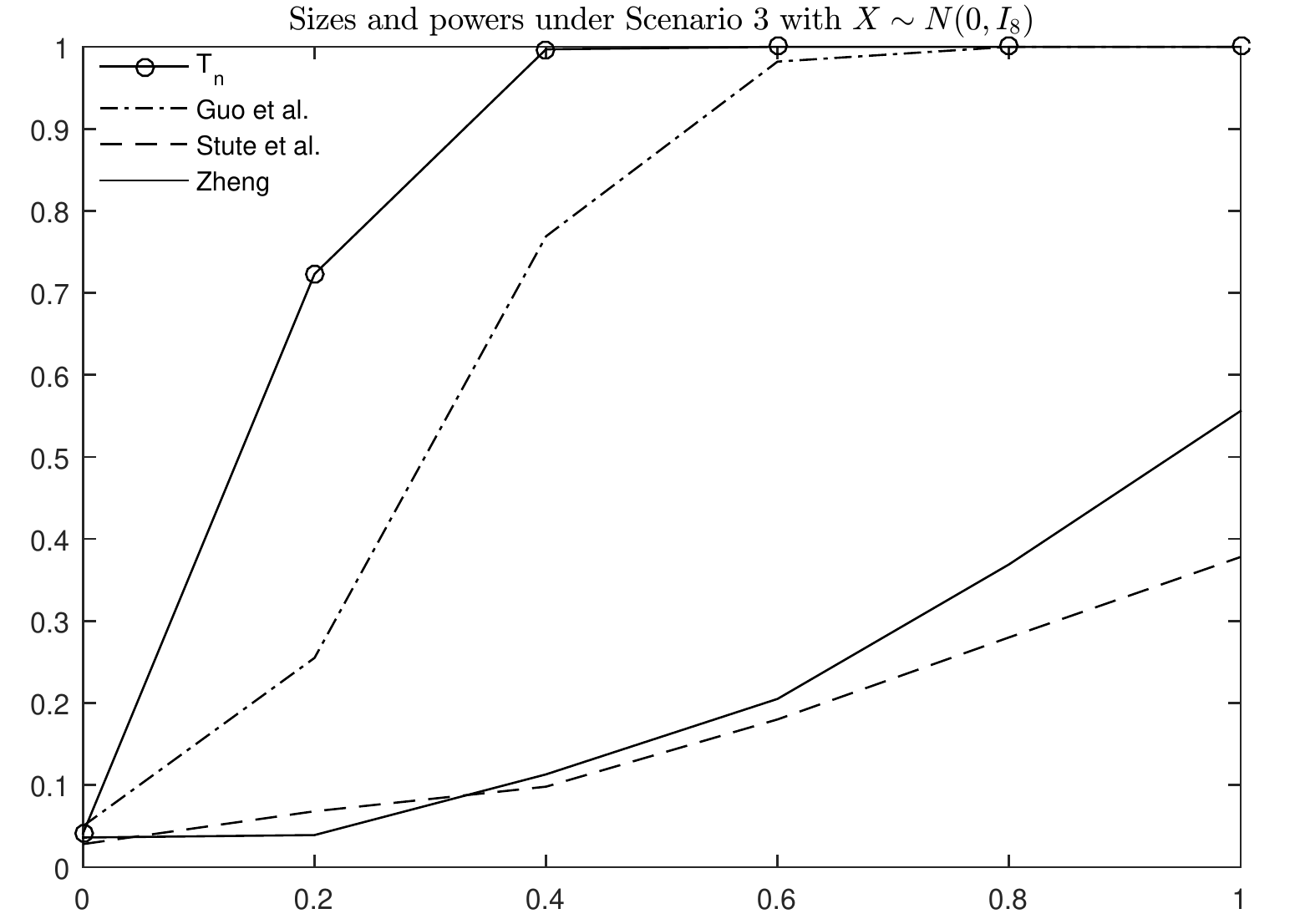}
 \caption{The empirical size and powers curves of $T_n$, $T^{GWZ}$, $T^S$ and $T^{Zh}$ in {\it Scenario 3} $n=200$ with $X\sim N(0,I_p)$. }
 \label{fig: scenario 3 independent}
\end{figure}

\begin{figure}
  \centering
  \includegraphics[width = 12cm]{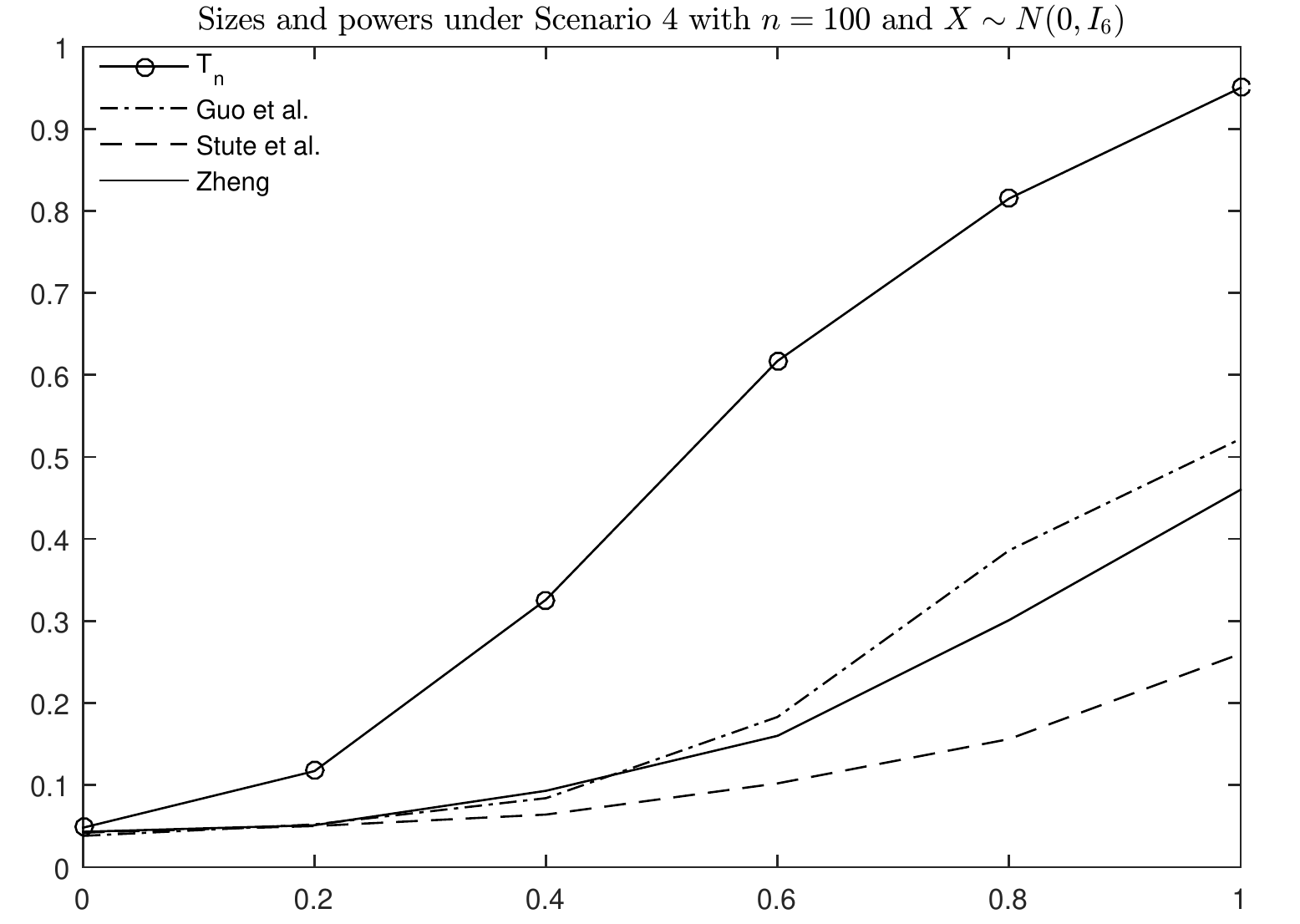}
  \includegraphics[width = 12cm]{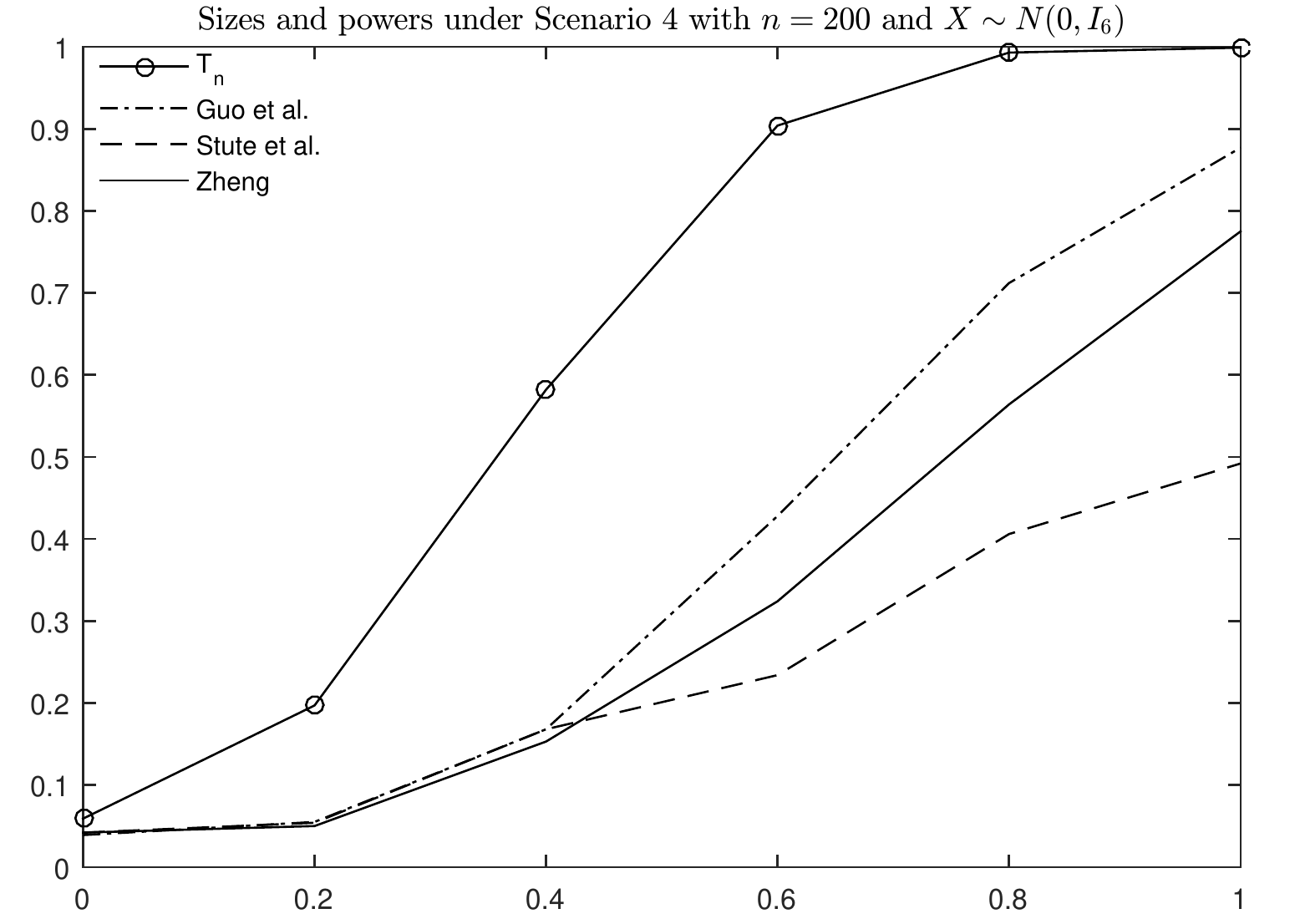}
  \caption{The empirical size and power curves of $T_n$, $T^{GWZ}$, $T^S$ and $T^{Zh}$ in {\it Scenario 4} with $n=100, 200$ and $X\sim N(0,I_6)$. }
  \label{fig: scenario 4}
\end{figure}

\end{document}